\documentclass[letterpaper, 10 pt, conference]{ieeeconf}
\IEEEoverridecommandlockouts
\usepackage{amsmath,amssymb,amsfonts}
\usepackage{amsthm}
\usepackage{algorithmic}
\usepackage{algorithm}
\usepackage{graphicx}
\usepackage{graphbox}
\usepackage{textcomp}
\usepackage{multirow}
\usepackage{balance}
\usepackage{xcolor}
\usepackage[export]{adjustbox}
\usepackage{color}
\usepackage{tikz}
\usepackage{silence}
\usepackage{accents}
\def\BibTeX{{\rm B\kern-.05em{\sc i\kern-.025em b}\kern-.08em
    T\kern-.1667em\lower.7ex\hbox{E}\kern-.125emX}}

\definecolor{LightCyan}{rgb}{0.85,0.85,1.0}
\newtheorem{theorem}{Theorem}

\newtheorem{lemma}{Lemma}

\newtheorem{definition}{Definition}
\newtheorem{proposition}{Proposition}

\newtheorem{problem}{Problem}

\makeatletter
\let\NAT@parse\undefined
\makeatother
\usepackage{hyperref}

\newcommand\copyrighttext{%
  \footnotesize \textcopyright This paper has been accepted for publication in the IEEE Conference on Decision and Control. Please cite the paper as: E. Sebasti\'{a}n, E. Montijano, and S. Sag\"{u}\'{e}s,``All-in-one: Certifiable Optimal Distributed Kalman Filter under Unknown Correlations'', IEEE Conference on Decision and Control (CDC), 6578-6583, 2021.}
\newcommand\copyrightnotice{%
\begin{tikzpicture}[remember picture,overlay]
\node[anchor=south,yshift=10pt] at (current page.south) {\fbox{\parbox{\dimexpr\textwidth-\fboxsep-\fboxrule\relax}{\copyrighttext}}};
\end{tikzpicture}%
}

\title{\LARGE\bf All-in-one: Certifiable Optimal Distributed Kalman Filter under Unknown Correlations}

\author{\centering Eduardo Sebasti\'{a}n, Eduardo Montijano and Carlos Sagü\'{e}s
\thanks{E. Sebasti\'{a}n, E. Montijano and C. Sagü\'{e}s are associated with the Instituto de Investigaci\'on en Ingenier\'ia de Arag\'on, Universidad de Zaragoza, Spain 
\texttt{\small \{esebastian, emonti, csagues\}@unizar.es}}
\thanks{This work was supported by ONRG-NICOP-grant N62909-19-1-2027, Spanish Projects PGC2018-098719-B-I00 (MCIU/AEI/FEDER, UE), DGA T04-FSE, and Spanish grant FPU19-05700. We are grateful for this support.}
}

\begin{document}
\maketitle

\copyrightnotice

\begin{abstract}
The optimal fusion of estimates in a Distributed Kalman Filter (DKF) requires tracking of the complete network error covariance, problematic in terms of memory and communication. A scalable alternative is to fuse estimates under unknown correlations, doing the update by solving an optimisation problem. Unfortunately, this problem is NP-hard, forcing relaxations that lose optimality guarantees. Motivated by this, 
we present the first Certifiable Optimal DKF (CO-DKF). Using only information from one-hop neighbours, CO-DKF solves the optimal fusion of estimates under unknown correlations by a particular tight Semidefinite Programming (SDP) relaxation which allows to certify, locally and in real time, if the relaxed solution is the actual optimum. In that case, we prove optimality in the Mean Square Error (MSE) sense. Additionally, we demonstrate the global asymptotic stability of the estimator.
CO-DKF outperforms other state-of-the-art DKF algorithms, specially in sparse, highly noisy setups.  
\end{abstract}


\section{Introduction}\label{sec:intro}

The Kalman Filter (KF)~\cite{Kalman1960Kalman} is a cornerstone in control theory. Its elegance and optimality has motivated its extension to distributed setups, with several promising applications such as cooperative tracking~\cite{Morbidi2012Active}
or monitoring~\cite{Khan2008Monitoring}.
However, the optimal fusion of estimates in DKF requires tracking of the whole network covariance~\cite{Olfati2009DKF}. This is impractical in terms of memory and communication.
The alternative is to consider unknown correlations, solving the estimation fusion sub-optimally or via an NP-hard optimisation problem.
The latest relies on convex relaxations, trading off formal guarantees of finding the original optimum. 
This is where \textit{certification} is of key importance, verifying if the relaxed solution is that of the original problem~\cite{Yang2020Teaser}. 

Exploiting this idea, our main contribution is the Certifiable Optimal Distributed Kalman Filter (CO-DKF), the \textit{first DKF with certifiability guarantees}.  
If the certification is positive, then CO-DKF is \textit{optimal in the Mean Square Error (MSE) sense under unknown correlations}.
Otherwise, it is still globally asymptotically stable.
Besides, nodes can check optimality locally and online, enabling security mechanisms to be exploited in applications where the estimator is embedded in the control structure~\cite{Talebi2019Distributed}.
In addition, CO-DKF is \textit{fully distributed}, has \textit{small communication burden} and works with \textit{heterogeneous} sensor models.


\section{Related work}\label{sec:SA}

The seminal works by Olfati-Saber~\cite{Olfati2005DKF,Olfati2007DKF} opened an era of consensus-based DKFs. The optimal gains in the MSE sense for the DKF proposed in~\cite{Olfati2007DKF} are found in~\cite{Olfati2009DKF}, but it is also proved that their computation is not scalable.
Since then, several works have searched for an Optimal DKF. 
In~\cite{Das2016Consensus} the authors propose a DKF composed by a dynamic consensus and a filtering step, where each agent knows the communication topology and sensor models. Two other DKFs are~\cite{Kamal2013Information,Battistelli2014Consensus}, where the consensus over information and measurements is exploited to improve the performance. The solution in~\cite{Kamal2013Information} converges to the centralised KF but reaching average consensus at each time step. These papers track the network error covariance matrix or do not consider the information in the error covariances to define optimality.

The absence of optimal \textit{scalable} solutions motivated the Diffusion KFs (DfKF). They assume \textit{unknown correlations} so that only local error covariances are leveraged, where the most popular approach is the Covariance Intersection (CI) method~\cite{Julier1997CI}. The work in~\cite{Reinhardt2015CI} shows that CI provides the optimal bound for two nodes, but for a greater number the CI method is generally not optimal, arising a necessity of certification.  
The authors of~\cite{Cattivelli2010Diffusion} proposed the DfKF suggesting a general convex optimisation. The CI method in DfKF was first applied in~\cite{Hu2011Diffusion}; however, arguing on its computational complexity, they study an approximation using only the estimate associated with the smaller error covariance trace. The work in~\cite{Deng2012sequential} proposes a sequential algorithm where each node performs a CI optimisation of two covariances, reducing the computation but increasing the communication burden. Two other CI-based DKFs are~\cite{wang2017convergence,he2018consistent}. The former provides novel proofs of stability for general time-varying and non-detectable scenarios, while the latter proposes an adaptive CI which decreases the conservativeness and preserves consistency. However, they do not guarantee optimality nor certifiability.

Inspired by the DfKF solutions, CO-DKF gathers the benefits of all the aforementioned solutions and adds the \textit{certifiability and optimality} properties, which come from a reformulation of the CI method based on the outer Löwner-John ellipsoid intersection method~\cite{Boyd1994LMI} and a SDP relaxation. Moreover, by proving the equivalence between CO-DKF and the optimal consensus over estimates we can demonstrate global asymptotic stability. An interesting consequence is an improved robustness against high levels of noise in the sensor measurements and sparse networks, validated in simulations that compare CO-DKF to other state-of-the-art filters.


\section{Problem formulation}\label{sec:preliminaries}

The target system is described by linear dynamics\footnote{For simplicity we do not consider external inputs and no matrix multiplying $\mathbf{w}(k)$. However, the results of the paper are directly generalisable.} 
\begin{equation}\label{eq:system}
    \mathbf{x}(k+1) = \mathbf{A}\mathbf{x}(k) + \mathbf{w}(k),
\end{equation}
where $k \in \mathbb{N}_{\ge 0}$ is the discrete time, $\mathbf{x} \in \mathbb{R}^{n}$ is the state, $\mathbf{A} \in \mathbb{R}^{n} \times \mathbb{R}^{n}$ is a matrix comprising the target system dynamics, and $\mathbf{w}(k) \sim \mathcal{N}(\mathbf{0},\mathbf{Q}(k))$ with $\hbox{E}[\mathbf{w}(a)\mathbf{w}^T(b)] := \mathbf{Q}(k) \delta_{ab}$ is the Zero Gaussian Mean (ZGM) state noise, where $\delta_{ab} = 1$ if $a=b$ and $\delta_{ab} = 0$ otherwise.

To track the system, there is a network of sensors with a communication graph $G = (V,E)$. The number of nodes is $N=|V|$. Moreover, $\mathcal{N}_i = \{j|(i,j)\in E \}$ is the set of neighbours of node $i$ and $\mathcal{J}_i = \mathcal{N}_i \cup i$ with $p = |\mathcal{J}_i|$ its cardinality. Each sensor is described by a linear model
\begin{equation}\label{eq:distributedSensorModel}
    \mathbf{z}_i(k) = \mathbf{H}_i\mathbf{x}(k) + \mathbf{v}_i(k).
\end{equation}
Here, $\mathbf{z}_i \in \mathbb{R}^{m_i}$ is the measurement of node $i$, $\mathbf{H}_i \in \mathbb{R}^{m_i} \times \mathbb{R}^{n}$ is the (unbiased) sensor model of node $i$, $m_i$ is the dimension of $\mathbf{z}_i$ and $\mathbf{v}_i(k) \sim \mathcal{N}(\mathbf{0},\mathbf{R}_i(k))$ with $\hbox{E}[\mathbf{v}_i(a)\mathbf{v}_i^T(b)] := \mathbf{R}_i(k) \delta_{ab}$ is the ZGM measurement noise, where $\delta_{ab} = 1$ if $a=b$ and $\delta_{ab} = 0$ otherwise. We consider independent measurements. Notice that by using $m_i$ we are considering heterogeneous and/or incomplete, as long as the system is observable by at least one node. 

The objective of the network is to cooperatively estimate $\mathbf{x}$. To do so, each node runs a local estimation algorithm, composed by a prediction step of the state ($\bar{\mathbf{x}}_i$), and an update of the prediction ($\hat{\mathbf{x}}_i$). 
Associated with the prediction there is an error covariance matrix $\mathbf{P}_i := \hbox{E}\{(\bar{\mathbf{x}}_i-\mathbf{x})(\bar{\mathbf{x}}_i-\mathbf{x})^T\} $. Since this matrix is not known, each node has a predicted ($\bar{\mathbf{P}}_{i}$) and an updated ($\hat{\mathbf{P}}_{i}$) version. Unless it is essential, in the paper we will omit the dependencies with $k$ for clearness.

Given complete knowledge of the network, the optimal fusion of these terms is
\begin{align}
    \bar{\mathbf{S}}^*_{i} &= (\bar{\mathbf{P}}_{i}^*)^{-1} = (\mathbf{1} \otimes \mathbf{I}_{n}) \bar{\mathbf{P}}_{c}^{-1} (\mathbf{1} \otimes \mathbf{I}_{n})^T, \label{eq:optimal1}
    \\
    \bar{\mathbf{x}}_{i}^* &= \bar{\mathbf{P}}^*_{i} (\mathbf{1} \otimes \mathbf{I}_{n}) \bar{\mathbf{P}}_{c}^{-1} [\bar{\mathbf{x}}_1^T, \hdots, \bar{\mathbf{x}}_N^T]^T,\label{eq:optimal2}
\end{align}
where $(\cdot)^*$ denotes optimality, $\mathbf{1}$ is the row vector of $N$ ones, $\mathbf{I}_{n}$ is the identity matrix of order $n$, $\otimes$ is the Kronecker product, and $\bar{\mathbf{P}}_{c}$ is the full covariance matrix of the network that includes the self-correlations $\bar{\mathbf{P}}_i$ and the cross-correlations $\bar{\mathbf{P}}_{ij}$. The optimality of~\eqref{eq:optimal1}-\eqref{eq:optimal2} is in terms of the MSE, defined as follows:
\begin{definition}[\textbf{Mean Square Error}]\label{definition:MSE}
The Mean-Square Error (MSE) is $\hbox{MSE} := \sum_{i=1}^{N} \hbox{E}[||\hat{\mathbf{x}}_i-\mathbf{x}||^2].$
\end{definition}
The optimal fusion in Eqs.~\eqref{eq:optimal1}-\eqref{eq:optimal2} is not scalable because it needs to track every $\bar{\mathbf{P}}_{ij},  \bar{\mathbf{P}}_{i}$ and $\bar{\mathbf{x}}_i$. This has a strong computational and memory burden, and requires communication with second order neighbours~\cite{Olfati2009DKF}. 

A typical alternative is to consider that the off-diagonal terms on $\bar{\mathbf{P}}_{c}$ are unknown and disregard non-neighbouring error covariances. Then, the optimal fusion can be achieved through optimisation, using the neighbouring predictions and predicted error covariances to produce an optimal bound of $\bar{\mathbf{P}}^*_i$ and the corresponding prediction. 
We rely on a bit abuse of notation and also denote this optimal bound by $\bar{\mathbf{P}}^*_i$, since it is the best result given unknown correlations and non-neighbouring error covariances.

This optimisation problem is intractable (Section~\ref{sec:solution}), which forces the use of convex relaxations to solve it. Despite enabling tractability, the relaxed solution is not guaranteed to be the original optimum. This motivates the need of certifiability on the optimisation, formally defined as follows:
\begin{definition}[\textbf{Certifiable Algorithms, from Definition 19 on~\cite{Yang2020Teaser}}]\label{definition:certifiability}
Given an optimization problem $\mathbb{O}(\mathbb{D})$ that depends on input data $\mathbb{D}$, we say that an algorithm $\mathbb{A}$ is certifiable if, after solving $\mathbb{O}(\mathbb{D})$, algorithm $\mathbb{A}$ either provides a certificate for the quality of its solution or declares failure otherwise.
\end{definition}

Now, we formulate the problem addressed in the paper.

\begin{problem}[\textbf{Certifiable optimal DKF under unknown correlations}]\label{problem:CO-DKF}
Find a stable algorithm that certificates locally and in real time if the estimation of system~\eqref{eq:system} carried out by each node is being optimal, i.e, at each $k$, each node $i$ minimises the MSE, under the following restrictions:
\begin{enumerate}
    \item \textbf{Unknown correlations}: The correlation terms $\mathbf{P_{ij}}(k)$ $\forall i \neq j$ are unknown $\forall i$ and $\forall k$.
    \item \textbf{Locality}: At instant $k$ and $\forall i$, node $i$ only uses $\mathbf{A}$, $\mathbf{H}_i$, $\mathbf{Q}(k)$, and $\mathbf{R}_i(k)$ as parameters. 
    \item \textbf{One-hop communication}: At instant $k$ and $\forall i$, node $i$ only communicates once with its neighbours $j \in \mathcal{N}_i$. 
\end{enumerate}
\end{problem}


\section{CO-DKF Algorithm}\label{sec:algorithm}

The proposed solution for Problem~\ref{problem:CO-DKF} is CO-DKF. We first offer an overview the algorithm and then, in Sections~\ref{sec:solution} and~\ref{sec:stability}, we formally study its main properties. 

At a given iteration of the algorithm, each node first sends and receives the prediction (with its covariance) and measurement (with its covariance) in information form to avoid the propagation of sensor models and unnecessary inversions of matrices,
\begin{equation*}
    \begin{aligned}
    \mathbf{u}_i &= \mathbf{H}_i^T\mathbf{R}_i^{-1}\mathbf{z}_i, &
    \mathbf{U}_i &= \mathbf{H}_i^T\mathbf{R}_i^{-1}\mathbf{H}_i,
    \\
    \bar{\mathbf{s}}_i &= \bar{\mathbf{P}}_i^{-1}\bar{\mathbf{x}}_i, & \bar{\mathbf{S}}_i &= \bar{\mathbf{P}}_i^{-1}.
    \end{aligned}
\end{equation*}
These quantities are aggregated using the information provided by the covariances. In the case of the measurements, this is direct because they are independent
\begin{flalign}
    \mathbf{y}_i = \frac{1}{p}\sum_{j\in \mathcal{J}_i} \mathbf{u}_j \hbox{ and }
    \mathbf{Y}_i = \frac{1}{p}\sum_{j\in \mathcal{J}_i} \mathbf{U}_j\label{eq:KF1}.
\end{flalign}
On the other hand, the aggregation of predictions is harder due to the correlations among them. Even if these correlations are unknown, they need to be taken into account into the fusion.
To solve it, we propose to use the outer Löwner-John (LJ) method~\cite{John2014Ellipsoid}, which gives the following relaxed optimisation problem  
\begin{subequations}\label{eq:outerLJellipse}
\begin{alignat}{2}
\bar{\mathbf{S}}_i^*, \lambda^* =  \:\:\: &\underset{\bar{\mathbf{S}},\lambda}{\arg\max}         \:\:\:\:\:\:\:\:\hbox{Tr}(\bar{\mathbf{S}})  \label{eq:optProb3}
\\
   \:\:\:\:\:\: s.t.  \:\:\:\:\:&\mathbf{0}  \prec \bar{\mathbf{S}} \preceq \sum_{j\in \mathcal{J}_i} \lambda_j \bar{\mathbf{S}}_j ,\label{eq:constraint31}
\\
&\sum_{j\in \mathcal{J}_i} \lambda_j   \leq 1 \hbox{ , } \lambda_j  \geq 0, \:\:\:\forall j \in \mathcal{J}_i \label{eq:constraint32}
\end{alignat}
\end{subequations}
where Tr$(\cdot)$ is the trace of a matrix, and $\prec$ and $\preceq$ denote definiteness and semidefiniteness. We recall that $\bar{\mathbf{P}}_i^{-1} = \bar{\mathbf{S}}_i$. The selection of the trace as the optimisation cost function follows from the definition of optimality, demonstrated later in the proof of Theorem~\ref{theorem:optimality}. The output of~\eqref{eq:outerLJellipse} is used to aggregate the predictions as follows
\begin{flalign}
    \bar{\mathbf{P}}^*_{i} = (\bar{\mathbf{S}}_i^*)^{-1} \hbox{ , }\:\:\:\:\:
    \bar{\mathbf{x}}_{i}^* = \bar{\mathbf{P}}_{i}^* \sum_{j\in \mathcal{J}_i} \lambda_j^*\bar{\mathbf{s}}_j\label{eq:KF_o1}.
\end{flalign}
The formal results over this aggregation are developed in Section~\ref{sec:solution}.
To the best of our knowledge, this is the first time optimisation problem~\eqref{eq:outerLJellipse} is applied to calculate $\bar{\mathbf{P}}^*_i$ and $\bar{\mathbf{x}}^*_i$ in a DKF. Furthermore, it enables the optimality certification of the CO-DKF.
With the data aggregated, each node calculates $\mathbf{M}_i:= \hbox{E}\{(\hat{\mathbf{x}}_i-\mathbf{x})(\hat{\mathbf{x}}_i-\mathbf{x})^T\}$ as
\begin{equation}
    \mathbf{M}_i = \left(\bar{\mathbf{S}}_i^* + \mathbf{Y}_i\right)^{-1}\label{eq:KF3}
\end{equation}
and then they correct the result of the predictions' aggregation by doing the following update
\begin{equation}\label{eq:CODKF_estimation}
    \hat{\mathbf{x}}_i^* = \bar{\mathbf{x}}_i^* + \mathbf{M}_i(\mathbf{y}_i - \mathbf{Y}_i\bar{\mathbf{x}}_i^*).
\end{equation}
It is noteworthy that the structure of Eq.~\eqref{eq:CODKF_estimation} is similar to 
\begin{equation}
    \label{eq:KF4}
    \hat{\mathbf{x}}_i =
        \bar{\mathbf{x}}_i + \mathbf{M}_i\left(\mathbf{y}_i - \mathbf{Y}_i\bar{\mathbf{x}}_i\right) +\gamma \mathbf{M}_i \sum_{j\in \mathcal{N}_i} (\bar{\mathbf{x}}_j - \bar{\mathbf{x}}_i),
\end{equation}
which comes from the consensus-based distributed Kalman filter in~\cite{Olfati2007DKF}. However, in CO-DKF the consensus is implicit in the optimisation, relaxing the dependence for stability on parameter $\gamma$. 
Finally, nodes predict for the next time-step,
\begin{flalign}
    \bar{\mathbf{P}}_i(k+1) &= \mathbf{A}\mathbf{M}_i(k)\mathbf{A}^T + \mathbf{Q}(k)\label{eq:KF5},
    \\
    \bar{\mathbf{x}}_i(k+1) &= \mathbf{A}\hat{\mathbf{x}}_i^*(k)\label{eq:KF6}.
\end{flalign}

We briefly discuss the communication and computational burden of this algorithm. Regarding communication, each node sends $\mathbf{U}_i, \mathbf{u}_i, \bar{\mathbf{S}}_i$ and $\bar{\mathbf{s}}_i.$
The size of the message is constant in the number of nodes, confirming the communication scalability of the proposal. Another advantage of CO-DKF is that nodes do not need any global knowledge of the topology nor sensor models of neighbours. Regarding computational cost, the bottleneck is the computation of~\eqref{eq:outerLJellipse}. The current hardware is able to solve large instances of this optimisation problem in real-time. 
This is discussed in more details in Section~\ref{sec:examples}.

In summary, CO-DKF is described by Algorithm~\ref{al:CO-DKF}.

\begin{algorithm}
\caption{CO-DKF in node $i$}\label{al:CO-DKF}
\begin{algorithmic}[1]
\STATE Initialisation: $\bar{\mathbf{P}}_i = \mathbf{P}_0$, $\bar{\mathbf{x}}_i = \mathbf{x}_0$
\WHILE{True}
    \STATE  Get $\mathbf{z}_i$, send $\{ \mathbf{U}_i, \mathbf{u}_i, \bar{\mathbf{S}}_i, \bar{\mathbf{s}}_i  \}$ and receive $\{ \mathbf{U}_j, \mathbf{u}_j, \bar{\mathbf{S}}_j, \bar{\mathbf{s}}_j  \}$ from neighbours.
    \STATE Aggregate measurements' data:
           \begin{itemize}
               \item[] $\mathbf{Y}_i = \frac{1}{p}\sum_{j\in \mathcal{J}_i} \mathbf{U}_j$ 
               \item[] $\mathbf{y}_i \:= \frac{1}{p}\sum_{j\in \mathcal{J}_i}  \mathbf{u}_j$ 
           \end{itemize}
    \STATE Aggregate predictions' data:
            \begin{itemize}
               \item[] \:$\bar{\mathbf{S}}_i^*$,$\mathbf{\lambda}^*$ $\leftarrow$ Solution of ~\eqref{eq:outerLJellipse}
               \item[] \:\:\:\:\:\:\:\:\:\:\:\:\:$\bar{\mathbf{P}}_{i}^* = (\bar{\mathbf{S}}_i^*)^{-1}$
               \item[] \:\:\:\:\:\:\:\:\:\:\:\:\:$\bar{\mathbf{x}}_{i}^* \:= \bar{\mathbf{P}}_{i}^* \sum_{j\in \mathcal{J}_i} \lambda_j^*\bar{\mathbf{s}}_j$
           \end{itemize}
    \STATE CO-DKF estimation update:
            \begin{itemize}
               \item[] $\mathbf{M}_i = (\bar{\mathbf{S}}^*_i + \mathbf{Y}_i)^{-1}$
               \item[] \:\:$\hat{\mathbf{x}}_i^* = \bar{\mathbf{x}}_{i}^* + \mathbf{M}_i(\mathbf{y}_i - \mathbf{Y}_i\bar{\mathbf{x}}_{i}^*)$
           \end{itemize}
    \STATE Compute the next-step prediction:
            \begin{itemize}
               \item[] $\bar{\mathbf{P}}_i(k+1) = \mathbf{A}\mathbf{M}_i(k)\mathbf{A}^T + \mathbf{Q}(k)$
               \item[] \:$\bar{\mathbf{x}}_i(k+1) = \mathbf{A}\hat{\mathbf{x}}_i^*(k)$
           \end{itemize}
\ENDWHILE 
\end{algorithmic}
\end{algorithm}


\section{Certifiable Covariance Bounding}\label{sec:solution}

To achieve optimality it is necessary to optimally aggregate the neighbouring predictions under unknown correlations. This fusion can be described as finding the minimum volume ellipsoid containing the intersection of the $p$ ellipsoids~\cite{Boyd1994LMI} formed by the matrices $\bar{\mathbf{P}}_j$ in $\mathcal{J}_i$. 
Thus, we first define the concept of ellipsoid and intersection of ellipsoids.
\begin{definition}[\textbf{Ellipsoid}]\label{definition:ellipsoid}
Given $\bar{\mathbf{S}}_i = (\bar{\mathbf{P}}_i)^{-1}$ and assuming unbiased sensors, the ellipsoid $\varepsilon_i$ is $\varepsilon_i := \{\,\mathbf{x} \,| \,\mathbf{x}^T \bar{\mathbf{S}}_i \mathbf{x} \leq 1 \, \}.$
\end{definition}
\begin{definition}[\textbf{Intersection of ellipsoids}]\label{definition:intersection}
The intersection of $p$ ellipsoids is the polytope 
$\mathcal{F} := \varepsilon_1 \cap \hdots \cap \varepsilon_p$.
\end{definition}

From Definitions~\ref{definition:ellipsoid}-\ref{definition:intersection}, minimising the volume can be transformed into a maximisation over the information matrices. In particular, since the MSE is equivalent to the trace of $\mathbf{M}_i,$ the objective of CO-DKF is to optimise this metric.

\begin{problem}[\textbf{Optimal intersection of ellipsoids}, adapted from~\cite{Boyd1994LMI}]\label{problem:intersection}
Find $\varepsilon^{*}_i$ such that contains $\mathcal{F}_i$ and \emph{Tr}$(\bar{\mathbf{S}}_i^*)$ is maximised.
\end{problem}
The solution of Problem~\ref{problem:intersection} is given by~\cite{Henrion2001LMI} 
\begin{subequations}\label{eq:traceProblem}
\begin{alignat}{2}
\underset{\mathbf{X}}{\max} & \:\:\:\:\:\:\:\hbox{Tr}(\mathbf{X}) \label{eq:optProb4}
\\
s.t. & \:\:\:\:\:\:\:    \hbox{Tr}(\mathbf{X} \bar{\mathbf{S}}_j)  \leq 1 \:\:\: \forall j \in \mathcal{J}_i,\label{eq:constraint41}
\\
&   \:\:\:\:\:\:\:   \mathbf{X}    \succeq \mathbf{0} \label{eq:constraint42}
\\
&   \:\:\:\:\:\:\:   \hbox{rank}(\mathbf{X})  = 1 \label{eq:constraint43}
\end{alignat}
\end{subequations}
Unfortunately, this problem is NP-hard because of constraint~\eqref{eq:constraint43}. Thus, it is necessary to find a convex relaxation. 
The simplest solution is to drop the non-convex constraint,
\begin{subequations}\label{eq:traceRelax}
\begin{alignat}{2}
\underset{\mathbf{X}}{\max} & \:\:\:\:\:\:\: \hbox{Tr}(\mathbf{X}) \label{eq:optProb5}
\\
s.t. & \:\:\:\:\:\:\:     \hbox{Tr}(\mathbf{X}\bar{\mathbf{S}}_j)  \leq 1 \:\:\: \forall j \in \mathcal{J}_i,\label{eq:constraint51}
\\
& \:\:\:\:\:\:\:\mathbf{X}    \succeq \mathbf{0} \label{eq:constraint52} 
\end{alignat}
\end{subequations}
However, once again it is not possible to use~\eqref{eq:traceRelax} in CO-DKF because it does not provide $\bar{\mathbf{S}}_i^*$ nor $\bar{\mathbf{x}}^*_i$, as is the case of the outer LJ relaxation in~\eqref{eq:outerLJellipse}.
Nevertheless, problem~\eqref{eq:traceRelax} is important because it enables the certification of optimality. This is formally stated in the next Proposition.

\begin{proposition}[\textbf{Certifiability}]\label{proposition:certifiability}
Let $\mathbf{X}^*$ be the solution of~\eqref{eq:traceRelax}.
Define
\begin{flalign}
    \mathbb{C}_i := \emph{rank}\left(\mathbf{X}^*\right) \hbox{ and }
    \rho_i &:= \emph{Tr}(\mathbf{X}^*) \vartheta(\bar{\mathbf{S}}_i^*)\in [0,1],
\end{flalign}
where $\vartheta(\bar{\mathbf{S}}_i^*)$ denotes the minimum eigenvalue of $\bar{\mathbf{S}}_i^*$, obtained solving~\eqref{eq:outerLJellipse}.
If $\mathbb{C}_i=\rho_i=1$, then the solution of~\eqref{eq:outerLJellipse} is the optimum of the original non-relaxed problem~\eqref{eq:traceProblem}.
\end{proposition}
\begin{proof}
If constraint~\eqref{eq:constraint43} holds for $\mathbf{X}^*$, then it is also the optimum of~\eqref{eq:traceProblem}. This can be demonstrated by contradiction. Consequently, $\mathbb{C}_i=1$ is a certificate of this equivalence.    

The next part is to show that the optimums of~\eqref{eq:traceRelax} and~\eqref{eq:outerLJellipse} are equivalent. This is described in detail in Section 2 of~\cite{Henrion2001LMI}, so we only sketch here the procedure. Obtained the Lagrangian of~\eqref{eq:traceRelax} and assuming that $\mathcal{F}$ does not reduce to zero, by Slater's condition both dual and primal have the same optimum. Then, by the S-procedure,~\eqref{eq:outerLJellipse} is equivalent to the dual and hence equivalent to~\eqref{eq:traceRelax}. 

If it is the case, then outputs of~\eqref{eq:traceRelax} and~\eqref{eq:outerLJellipse} can be compared to see how close is the solution of~\eqref{eq:outerLJellipse} to the solution of Problem~\ref{eq:traceProblem}. To do so, Theorem 4 of~\cite{Henrion2001LMI} states that the optimal value of~\eqref{eq:outerLJellipse} is bounded by $1/\vartheta(\bar{\mathbf{S}}^*_i)$, where $\bar{\mathbf{S}}^*_i$ is the output of the optimisation in~\eqref{eq:outerLJellipse}. Therefore, $\rho_i$ assesses the output of~\eqref{eq:outerLJellipse} comparing it with the optimal value certified by the constraint~\eqref{eq:constraint43}. If $\mathbb{C}_i=1$ and $\rho_i=1$, then the solution of~\eqref{eq:outerLJellipse} is the optimum of the original Problem~\ref{problem:intersection}.  
\end{proof}

Thus, the proposition gives a procedure to check, locally and online, the optimisation process by finding the optimal value for Problem~\ref{problem:intersection}. 


\section{Stability and Optimality}\label{sec:stability}

In this Section we analyse the global properties of CO-DKF. To this end, we define the update error as $\eta_i = \hat{\mathbf{x}}_i-{\mathbf{x}}$, 
\begin{equation*}
    \begin{aligned}
        \eta :=& [\eta_1^T, \hdots, \eta_{N}^T]^T,  & \mathbf{H} :=& \hbox{block-diag}(\mathbf{H}_1, \hdots, \mathbf{H}_{N})
        \\
        \mathbf{z} :=& [\mathbf{z}_1^T, \hdots, \mathbf{z}_N^T]^T, &\mathbf{R} :=& \hbox{block-diag}(\mathbf{R}_1, \hdots, \mathbf{R}_{N}),
        \\
        \mathbf{y} :=& [\mathbf{y}_1^T, \hdots, \mathbf{y}_N^T]^T,&\mathbf{U} :=& \hbox{block-diag}(\mathbf{U}_1, \hdots, \mathbf{U}_{N}),
        \\
        \mathbf{\mathcal{A}} :=& \mathbf{I} \otimes \mathbf{A}, &\mathbf{M} :=& \hbox{block-diag}(\mathbf{M}_1, \hdots, \mathbf{M}_{N}),
        \\
        \mathbf{\mathcal{Q}} :=& \mathbf{I} \otimes \mathbf{Q},&\bar{\mathbf{P}}^* :=& \hbox{block-diag}(\bar{\mathbf{P}}_{1}^*, \hdots, \bar{\mathbf{P}}_{N}^*),
    \end{aligned}
\end{equation*}
and
\begin{equation}\label{eq:nw}
\begin{aligned}
    \mathbf{Y} \overset{Eq.~\eqref{eq:KF1}}{=} \mathbf{N}_w \mathbf{U} \Rightarrow & 
    \left. \begin{array}{ll}
         |\mathbf{N}_w|_{ij} = \frac{1}{p}*\mathbf{I}  \hbox{ if } j\in \mathcal{J}_i  \\
         |\mathbf{N}_w|_{ij} = 0*\mathbf{I} \hbox{ otherwise}
    \end{array}\right..
\end{aligned}
\end{equation}

To demonstrate the global asymptotic stability of CO-DKF, we adapt two Lemmas from~\cite{Olfati2009DKF} and derive a novel one which will support our main result. 
\begin{lemma}[\textbf{Adapted from Lemma 2 from~\cite{Olfati2009DKF}}]\label{lemma:matrices}
Given Eqs.~\eqref{eq:KF3}-\eqref{eq:CODKF_estimation}, the following holds:
\begin{enumerate}
    \item $\mathbf{F} = \mathbf{I} - \mathbf{M}\mathbf{Y}=\mathbf{M}(\bar{\mathbf{P}}^*)^{-1}$.
    \item $\mathbf{M} = \mathbf{F}\mathbf{G}\mathbf{F}^T$ with $\mathbf{G} = \mathbf{\mathcal{A}}\mathbf{M}\mathbf{\mathcal{A}}^T + \mathbf{\mathcal{Q}} + \mathbf{T}\mathbf{R}^{-1}\mathbf{T}^T$ and $\mathbf{T} = \bar{\mathbf{P}}^* \mathbf{N}_w \mathbf{H}^T$.
\end{enumerate}
\end{lemma}
\begin{proof}
Let $\mathbf{F} = \mathbf{I} - \mathbf{M}\mathbf{Y}$, by Eq.~\eqref{eq:KF3} we can write $\mathbf{M} = ((\bar{\mathbf{P}}^*)^{-1} + \mathbf{Y})^{-1}$. Then,
\begin{equation}
  \mathbf{M}((\bar{\mathbf{P}}^*)^{-1} + \mathbf{Y}) = \mathbf{I} \Rightarrow \mathbf{I} - \mathbf{M}\mathbf{Y}=\mathbf{M}(\bar{\mathbf{P}}^*)^{-1}
\end{equation}
and the first statement is proved. Regarding the second statement, notice that Eq.~\eqref{eq:CODKF_estimation} can be stacked as follows
\begin{equation}\label{eq:CODKF_estimation_stacked}
    \hat{\mathbf{x}}^* = \bar{\mathbf{x}}^* + \mathbf{M}(\mathbf{y} - \mathbf{Y}\bar{\mathbf{x}}^*).
\end{equation}
Then, taking into account that $\mathbf{y}=\mathbf{N}_w\mathbf{H}^T\mathbf{R}^{-1}\mathbf{z}$ and $\mathbf{Y}=\mathbf{N}_w\mathbf{U} = \mathbf{N}_w\mathbf{H}^T\mathbf{R}^{-1}\mathbf{H}$, Eq.~\eqref{eq:CODKF_estimation_stacked} is rewritten from its information form as
\begin{equation}\label{eq:CODKF_estimation_standard}
    \hat{\mathbf{x}}^* = \bar{\mathbf{x}}^* + \mathbf{M}\mathbf{N}_w\mathbf{H}^T\mathbf{R}^{-1}(\mathbf{z} - \mathbf{H}\bar{\mathbf{x}}^*) = \bar{\mathbf{x}}^* + \mathbf{K}(\mathbf{z} - \mathbf{H}\bar{\mathbf{x}}^*) 
\end{equation}
with $\mathbf{K}=\mathbf{M}\mathbf{N}_w\mathbf{H}^T\mathbf{R}^{-1}$ the Kalman Gain of the standard KF formulation. Let now consider the update of $\mathbf{M}$ as
\begin{equation}\label{eq:standard}
    \mathbf{M}^{+} = \mathbf{F}\mathbf{P}^+\mathbf{F}^T + \mathbf{K}\mathbf{R}\mathbf{K}^T,    
\end{equation}
where $(\cdot)^+$ is the update operator. This expression comes from the standard KF~\cite{Olfati2009DKF}. Therefore, substituting $\mathbf{K}$ we obtain
\begin{equation}\label{eq:standard2}
    \mathbf{M}^{+} = \mathbf{F}(\mathbf{\mathcal{A}}\mathbf{M}\mathbf{\mathcal{A}}^T + \mathbf{\mathcal{Q}})\mathbf{F}^T + \mathbf{M}\mathbf{N}_w\mathbf{H}^T\mathbf{R}^{-1}\mathbf{H}\mathbf{N}_w^T\mathbf{M}^T.   
\end{equation}
Given statement 1 of the Lemma,
\begin{equation}\label{eq:standard3}
    \mathbf{M}^{+} = \mathbf{F}(\mathbf{\mathcal{A}}\mathbf{M}\mathbf{\mathcal{A}}^T + \mathbf{\mathcal{Q}} + \bar{\mathbf{P}}^*\mathbf{N}_w\mathbf{H}^T\mathbf{R}^{-1}\mathbf{H}\mathbf{N}_w^T(\bar{\mathbf{P}}^*)^T)\mathbf{F}^T,   
\end{equation}
and statement 2 is proved.
\end{proof}
\begin{lemma}[\textbf{Adapted from Lemma 3 from~\cite{Olfati2009DKF}}]\label{lemma:stability}
Suppose that the error dynamics without noise are $\eta^+ = \mathbf{F}\mathbf{\mathcal{A}}\eta$ with $\mathbf{F}$ defined as in Lemma~\ref{lemma:matrices}. Then, the error dynamics is globally asymptotically stable system with a Lyapunov function $V(\eta) = \eta^T \mathbf{M}^{-1} \eta$ and $\Lambda = \mathbf{M}^{-1} - \mathbf{\mathcal{A}}^T\mathbf{G}^{-1} \mathbf{\mathcal{A}} \succ \mathbf{0}$.
\end{lemma}
\begin{proof}
Given $V(\eta) = \eta^T \mathbf{M}^{-1} \eta$ as Lyapunov function candidate,
\begin{equation}
\begin{aligned}
    \delta V =& (\eta^+)^T (\mathbf{M}^+)^{-1} \eta^+ - \eta^T \mathbf{M}^{-1} \eta =\\
              & \eta^T\left.(\mathbf{\mathcal{A}}^T\mathbf{F}^T (\mathbf{M}^+)^{-1} \mathbf{F} \mathbf{\mathcal{A}} - \mathbf{M}^{-1})\right. \eta = \\
              & \eta^T\left.(\mathbf{\mathcal{A}}^T\mathbf{G}^{-1} \mathbf{\mathcal{A}} - \mathbf{M}^{-1})\right. \eta = \\
              & -\eta^T\left.(\mathbf{M}^{-1} - \mathbf{\mathcal{A}}^T(\mathbf{\mathcal{A}}\mathbf{M}\mathbf{\mathcal{A}}^T + \mathbf{W})^{-1} \mathbf{\mathcal{A}})\right. \eta = \\
              & -\eta^T \Lambda \eta,
\end{aligned}
\end{equation}
with $\mathbf{W} = \mathbf{\mathcal{Q}} + \mathbf{T}\mathbf{R}^{-1}\mathbf{T}^T \succ \mathbf{0}$ and $\Lambda = \mathbf{M}^{-1} - \mathbf{\mathcal{A}}^T(\mathbf{\mathcal{A}}\mathbf{M}\mathbf{\mathcal{A}}^T + \mathbf{W})^{-1}\mathbf{\mathcal{A}}$. The rest of the proof directly follows from Lemma 3 of~\cite{Olfati2009DKF}, showing that $\Lambda \succ \mathbf{0}$. Therefore, the error dynamics is globally asymptotically stable.
\end{proof}
Finally, to prove stability of CO-DKF we provide a Lemma associated to~\eqref{eq:outerLJellipse}.  
show that, if the predicted covariance matrices converge to a constant value, then  constraint~\eqref{eq:constraint31} in an strict equality.
\begin{lemma}[\textbf{Strong equalities}]
\label{lemma:convergence}
Given optimisation problem~\eqref{eq:outerLJellipse}, $\bar{\mathbf{S}}^*_i = \sum_{j=1}^p \lambda_j^* \bar{\mathbf{S}}_j$ with $j \in \mathcal{J}_i$, for all $i$. 
\end{lemma}
\begin{proof}
Problem~\eqref{eq:outerLJellipse} is a convex maximisation with respect to $\bar{\mathbf{S}}$ and $\lambda$, which means that its optimum is unique. Notice that the cost function is the trace of $\bar{\mathbf{S}}$. Therefore, maximisation of Tr$(\bar{\mathbf{S}})$ with respect to $\lambda$ implies that the maximum is found when constraint~\eqref{eq:constraint32} is an equality and $\sum_{j=1}^p \lambda_j^* \bar{\mathbf{S}}_j$ is the convex combination of $\bar{\mathbf{S}}_j$ with maximum trace. This can be demonstrated by contradiction: if $\sum_{j=1}^p \lambda_j < 1$ then there exists another convex combination of $\bar{\mathbf{S}}_j$ with a higher trace, which means that the optimum is with $\sum_{j=1}^p \lambda_j = 1$; in addition, the values of $\lambda_j$ must be such that $\sum_{j=1}^p \lambda_j \bar{\mathbf{S}}_j$ is maximised, otherwise there exists a different convex combination of $\bar{\mathbf{S}}_j$ with a higher trace. 

Now, notice that $\sum_{j=1}^p \lambda_j^* \bar{\mathbf{S}}_j$ is in constraint~\eqref{eq:constraint31} bounding the value of $\bar{\mathbf{S}}$ such that matrix $\sum_{j=1}^p \lambda_j^* \bar{\mathbf{S}}_j - \bar{\mathbf{S}}$ must be positive semidefinite. Therefore, the optimal value of problem~\eqref{eq:outerLJellipse} must be the supremum of constraint~\eqref{eq:constraint31} and the optimal value of $\bar{\mathbf{S}}$ must share eigenvalues with $\sum_{j=1}^p \lambda_j^* \bar{\mathbf{S}}_j$. Since they share eigenvalues, the trace of both quantities are equal and $\hbox{Tr}(\bar{\mathbf{S}}^*_i) = \hbox{Tr}(\sum_{j=1}^p \lambda_j^* \bar{\mathbf{S}}_j)$ holds. Given the uniqueness of optimum, the optimal value of $\bar{\mathbf{S}}$ is $\sum_{j=1}^p \lambda_j^* \bar{\mathbf{S}}_j$ and the statement of the Lemma is proved.

\end{proof}

In essence, Lemma~\ref{lemma:convergence} says that the optimisation in~\eqref{eq:outerLJellipse} becomes a standard discrete-time consensus protocol tuned to optimise the trace of the final consensus value of each node. Interestingly, from~\eqref{eq:KF3} we can expect that the optimisation will ponder more the nodes equipped with better sensors, which is a positive side effect of CO-DKF when compared with the widely used update equation in~\eqref{eq:KF4}.

We can now demonstrate the stability of the estimator.
\begin{theorem}[\textbf{Stability}]\label{theorem:stability}
CO-DKF is a globally asymptotically stable estimator.
\end{theorem}
\begin{proof}
The first step is to rewrite Eq.~\eqref{eq:CODKF_estimation} like Eq.~\eqref{eq:KF4}. The aggregated prediction $\bar{\mathbf{x}}^*_{i}$ in Eq.~\eqref{eq:CODKF_estimation} is obtained by
\begin{equation}\label{eq:xfusiones}
     \bar{\mathbf{x}}_{i}^* = \bar{\mathbf{P}}_{i}^* \sum_{j\in \mathcal{J}_i} \lambda_j^*\bar{\mathbf{P}}_j^{-1}\bar{\mathbf{x}}_{j}.
\end{equation}
This expression can be rewritten as
\begin{equation}\label{eq:xfusiones3}
     \bar{\mathbf{x}}_{i}^* = \bar{\mathbf{P}}_{i}^*\sum_{j\in \mathcal{J}_i} \lambda_j^*\bar{\mathbf{P}}_j^{-1}  \bar{\mathbf{x}}_{i}  +
     \bar{\mathbf{P}}^*_{i} \sum_{j\in \mathcal{N}_i} \lambda_j^*\bar{\mathbf{P}}_j^{-1}(\bar{\mathbf{x}}_{j}-\bar{\mathbf{x}}_{i}).
\end{equation}
From Lemma~\ref{lemma:convergence} we have that $\bar{\mathbf{P}}_{i}^*\sum_{j\in \mathcal{J}_i} \lambda_j^*\bar{\mathbf{P}}_j^{-1} = \mathbf{I}$. Thus, plugging~\eqref{eq:xfusiones3} into~\eqref{eq:CODKF_estimation} gives
\begin{equation}\label{eq:rewrite_sta}
\begin{aligned}
     \hat{\mathbf{x}}_{i} =& \bar{\mathbf{x}}_{i} + \mathbf{M}_i(\mathbf{y}_i - \mathbf{Y}_i\bar{\mathbf{x}}_{i})  + 
     \\&
     (\mathbf{I}-\mathbf{M}_i\mathbf{Y}_i)\bar{\mathbf{P}}_{i}^*
     \sum_{j\in \mathcal{N}_i} \lambda_j^*\bar{\mathbf{P}}_j^{-1}(\bar{\mathbf{x}}_{j}-\bar{\mathbf{x}}_{i}).
\end{aligned}
\end{equation}
The last expression is equivalent to Eq.~\eqref{eq:KF4} but weighting each term with the result of optimising~\eqref{eq:outerLJellipse}.
Given this equivalence, the noiseless dynamics of $\eta_i$ is
\begin{equation}\label{eq:error_dynamics}
     \eta_{i}^+ = (\mathbf{I}-\mathbf{M}_i\mathbf{Y}_i)(\mathbf{A}\eta_{i} + 
     \bar{\mathbf{P}}_{i}^*\sum_{j\in \mathcal{N}_i} \lambda_j^*\bar{\mathbf{P}}_j^{-1}\mathbf{A}(\eta_{j}-\eta_{i})),
\end{equation}
which can be written in compact form as $\eta^+ = \mathbf{F}\left.\mathbf{L}_w\right. \mathbf{\mathcal{A}}\eta.$ Here, $(\cdot)^+$ is the update operator, $\mathbf{F}$ is defined as in Lemma~\ref{lemma:matrices}, and $\mathbf{L}_w$ is such that $|\mathbf{L}_w|_{ij} = \bar{\mathbf{P}}_{i}^*\lambda_j^*\bar{\mathbf{P}}_j^{-1}$ for all $j\in \mathcal{J}_i$ and $\mathbf{0}$ otherwise. The latest is similar to the error dynamics proved as globally asymptotically stable in Lemma~\ref{lemma:stability}, but with $\mathbf{L}_w$ in between. Lemma~\ref{lemma:stability} states that $\Lambda = \mathbf{M}^{-1} - \mathbf{\mathcal{A}}^T\mathbf{G}^{-1} \mathbf{\mathcal{A}} \succ \mathbf{0}$
but, since the error dynamics is $\eta^+$, we instead have $\Lambda^{'} = \mathbf{M}^{-1} - \mathbf{\mathcal{A}}^T\mathbf{L}_w^T\mathbf{G}^{-1}\mathbf{L}_w\mathbf{\mathcal{A}}.$ If 
\begin{equation}\label{eq:statement1}
 \mathbf{\mathcal{A}}^T\mathbf{G}^{-1} \mathbf{\mathcal{A}} \succeq \mathbf{\mathcal{A}}^T\mathbf{L}_w^T\mathbf{G}^{-1}\mathbf{L}_w\mathbf{\mathcal{A}}, 
\end{equation}
then $\Lambda^{'} \succ \mathbf{0}$ and global asymptotic stability of the error dynamics is proved. The statement in~\eqref{eq:statement1} is equivalent to $\mathbf{I} \succeq \mathbf{L}_w$. Given Lemma~\ref{lemma:convergence}, $\mathbf{L}_w$ is a row-stochastic matrix. Therefore, by linear algebra results, the absolute value of any eigenvalue of $\mathbf{L}_w$ is less than or equal to $1$. This means that the eigenvalues of matrix $\mathbf{L}_w - \mathbf{I}$ are all negative or equal to $0$ (is negative semidefinite) and, then, $\mathbf{0} \succeq \mathbf{L}_w - \mathbf{I}$.

Thus, the statement $\mathbf{I} \succeq \mathbf{L}_w$ holds and $\Lambda^{'} \succ \mathbf{0}$.
\end{proof}

As a remark, notice that the stability result applies to CO-DKF independently of the certification. 
The next step is to demonstrate that CO-DKF is optimal when it is certified. 

\begin{theorem}[\textbf{Optimality}]\label{theorem:optimality}
Assume that the solution of~\eqref{eq:outerLJellipse} is certified as optimal. Then, it provides the optimal consensus gains of Eq.~\eqref{eq:KF4}, under the restrictions of Problem~\ref{problem:CO-DKF}.
\end{theorem}
\begin{proof}
The expression in Eq.~\eqref{eq:rewrite_sta} can be reformulated as
\begin{equation}\label{eq:rewrite}
     \hat{\mathbf{x}}_{i}^* = \bar{\mathbf{x}}_{i} + \mathbf{M}_i(\mathbf{y}_i - \mathbf{Y}_i\bar{\mathbf{x}}_{i})  + 
     \gamma\mathbf{M}_i\sum_{j\in \mathcal{N}_i} \lambda_j^*\bar{\mathbf{P}}_j^{-1}(\bar{\mathbf{x}}_{j}-\bar{\mathbf{x}}_{i})
\end{equation}
with $\gamma = (\mathbf{I}-\mathbf{M}_i\mathbf{Y}_i)\bar{\mathbf{P}}_{i}^*\mathbf{M}_i^{-1}$. The latest is equivalent to Eq.~\eqref{eq:KF4} but weighting each term with the result of optimising~\eqref{eq:outerLJellipse}. Recalling Definition~\ref{definition:MSE}, the MSE is $\sum_i^{N}\hbox{E}[||\hat{\mathbf{x}}_i-\mathbf{x}||^2]$, which can be rewritten as $\sum_i^{N}\hbox{Tr}(\mathbf{M}_i)$. Since $\mathbf{M}_i = (\bar{\mathbf{S}}_{i}^{*} + \mathbf{Y}_i)^{-1}$ and $\bar{\mathbf{S}}_{i}^{*}$ comes from the optimisation in~\eqref{eq:outerLJellipse} which optimises over the trace of $\bar{\mathbf{S}}_{i}$, and this solution is certified as optimal, then the solution of~\eqref{eq:outerLJellipse} is equivalent to finding the optimal consensus gains of Eq.~\eqref{eq:KF4}.
\end{proof}

This result means that the optimisation carried on to fuse the neighbouring predictions is equivalent to finding the optimal consensus gains in the MSE sense. Thus, if the node certifies its optimality at instant $k$, then the estimation at instant $k$ is optimal. 

It is interesting to note the differences of CO-DKF with~\cite{wang2017convergence,he2018consistent}. Despite taking a similar approach for the fusion of estimates, they both omit, implicitly or explicitly, the consensus over estimates and/or measurements. Moreover, the optimisation problem for the fusion of estimates is the common CI (the adaptive CI in~\cite{he2018consistent} is another optimisation problem over the original one). These two aspects are key differential because (i) our reformulation allows to certify optimality in the weights' election and therefore, (ii) in conjunction with the  explicit consensus over measurements and implicit consensus over estimates, we can proof optimality.


\section{Simulations}\label{sec:examples}

To validate CO-DKF we compare the performance of different estimators: (i) \textit{CO-DKF}, our proposal, (ii) \textit{CDFK}, Algorithm 3 from~\cite{Olfati2007DKF} with consensus gain $\gamma_i = 1/(||\mathbf{M}_i||_{F} + 1)$ and $||\cdot||_F$ the Frobenius norm, (iii) \textit{HDfKF}, complete algorithm in~\cite{Hu2011Diffusion}, (iv) \textit{HADfKF}, simplified algorithm in~\cite{Hu2011Diffusion}, (v) \textit{CKF}, centralised equivalent KF. The target system is a 2D particle described by 
\begin{align}
    \mathbf{x} &= \begin{pmatrix}
    x & y & v_x & v_y
    \end{pmatrix}^T 
    \\
    \mathbf{A} &= 
    \begin{pmatrix}
    1 & 0 & \sin{w_p T} &  \cos{w_p T} - 1
    \\ 
    0 & 1 & 1-\cos{w_p f T} & \sin{w_p T}
    \\ 
    0 & 0 & \cos{w_p T} & \sin{w_p T} 
    \\ 
    0 & 0 & \sin{w_p T}  & \cos{w_p T}
    \end{pmatrix} 
\end{align}
with $w_p=0.5$rad/s, $T=0.1$s and $\mathbf{Q}=2\times 10^{-6}\mathbf{I}_4$. 

To do the comparison, we analyse two scenarios. Experiment $1$ instantiates a random Laplacian to produce a random sensor network, with appropriate parameters to obtain a sparse connected topology. Then, we uniformly randomly decide the quantities each sensors measures, selecting among measuring $x$, measuring $y$ or measuring $x$ and $y$. Notice that we never measure $v_x$ nor $v_y$. Besides, the values of each $\mathbf{H}_i$ come from a uniform distribution in the range $[1,3]$. After that, a $p=0.5$ Bernoulli distribution decides if the diagonal values of $\mathbf{R}_i$ are in the range $[3,5]\times 10^{-2}$ or $[3,5]$, i.e., if the sensor has high or low quality. This process is done 100 times, computing the MSE and doing the average over the experiments, which is the same procedure in~\cite{Cattivelli2010Diffusion} and~\cite{Hu2011Diffusion}. Experiment $2$ is more challenging and aims at validating the robustness and performance of CO-DKF. We repeat the initialisation process of Experiment $1$ but only one sensor is of high quality. An instance is shown in Fig.~\ref{fig:topology}.
\begin{figure}[!ht]
\centering
\includegraphics[width=0.8\columnwidth]{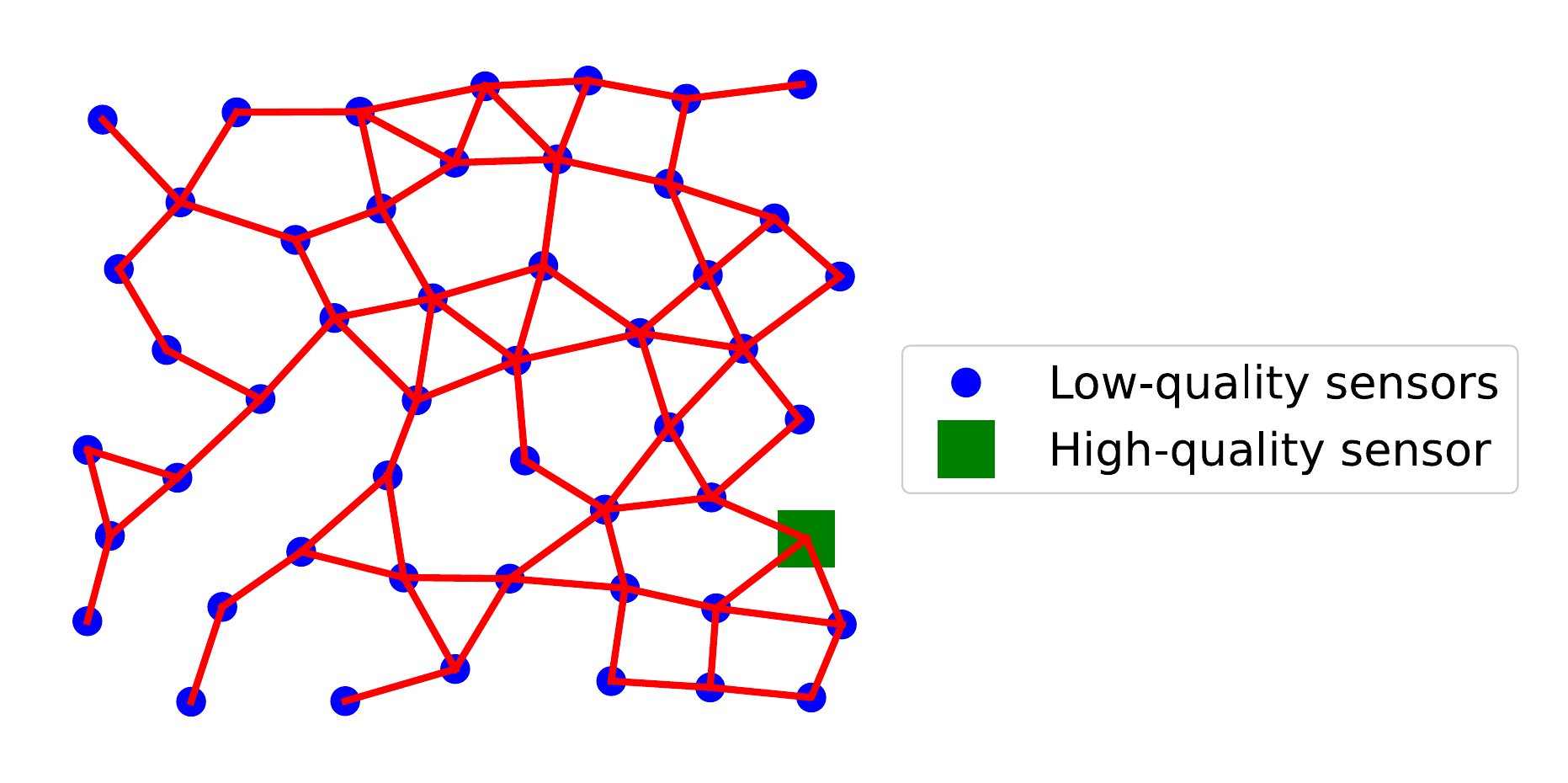}
\caption{Example of random sensor network generated in Experiment $2$.}
\label{fig:topology}
\end{figure}

In the case of the CDKF, we have computed the average over the successful experiments, which is a $42\%$ in Experiment $1$ and a $15\%$ in Experiment $2$. This is reasonable since the CDKF needs a sufficiently small consensus gain to be stable~\cite{Olfati2009DKF}, subject to the quality and density of the measurements. For the rest of the estimators, all the experiments converged. 

The results of Experiment $1$ are shown in the first column of Fig.~\ref{fig:RMSEvsREAL_transient}. The best performance is obtained by CO-DKF, with a difference of an order of magnitude with HDfKF and HADfKF, and of several orders of magnitude with CDKF. Indeed, the trace of the CDKF only appears in the zoomed pannel in Fig.~\ref{fig:RMSEvsREAL_transient}. Focusing near the steady-state, we can see that CO-DKF is also the best, with the fastest convergence to the CKF and with a difference of an order of magnitude with HDfKF/HADfKF. The CDKF improves its performance, without surpassing CO-DKF, but we recall that it corresponds to a $42\%$ of success rate. 

The differences are greater in Experiment $2$, as it is shown in the second column of Fig.~\ref{fig:RMSEvsREAL_transient}, where the HADfKF diverges and HDfKF exhibits a much worse transient response. Therefore, CO-DKF obtains the best performance, coming with certifiability, optimality and stability guarantees, and with no need of tuning even with only a single high-quality sensor.

\begin{figure}[!ht]
\centering
\begin{tabular}{cc}
    {\footnotesize Experiment 1}
    &
    {\footnotesize Experiment 2}
    \\
    \includegraphics[width=0.45\columnwidth,height=0.4\columnwidth]{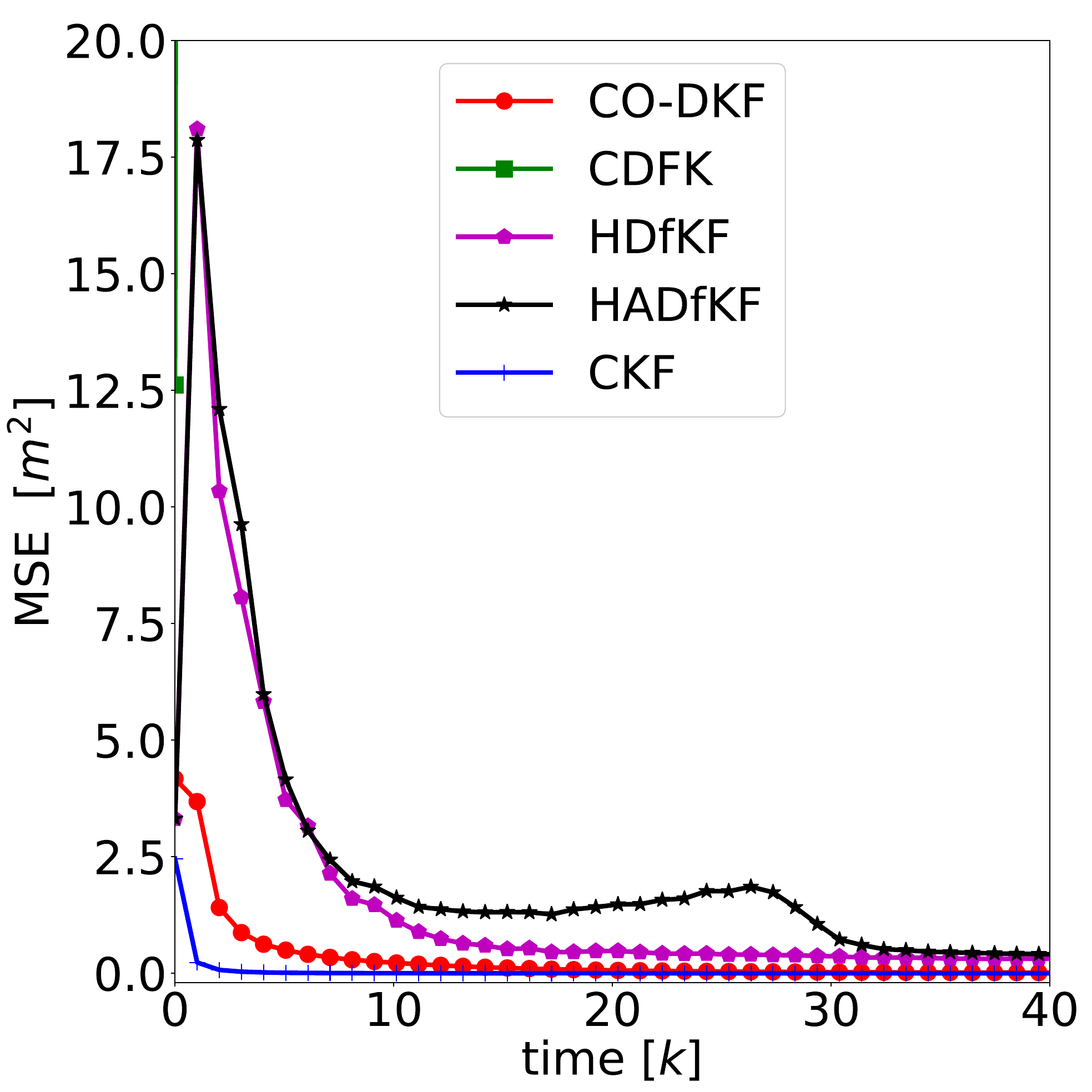}
    &
    \includegraphics[width=0.45\columnwidth,height=0.4\columnwidth]{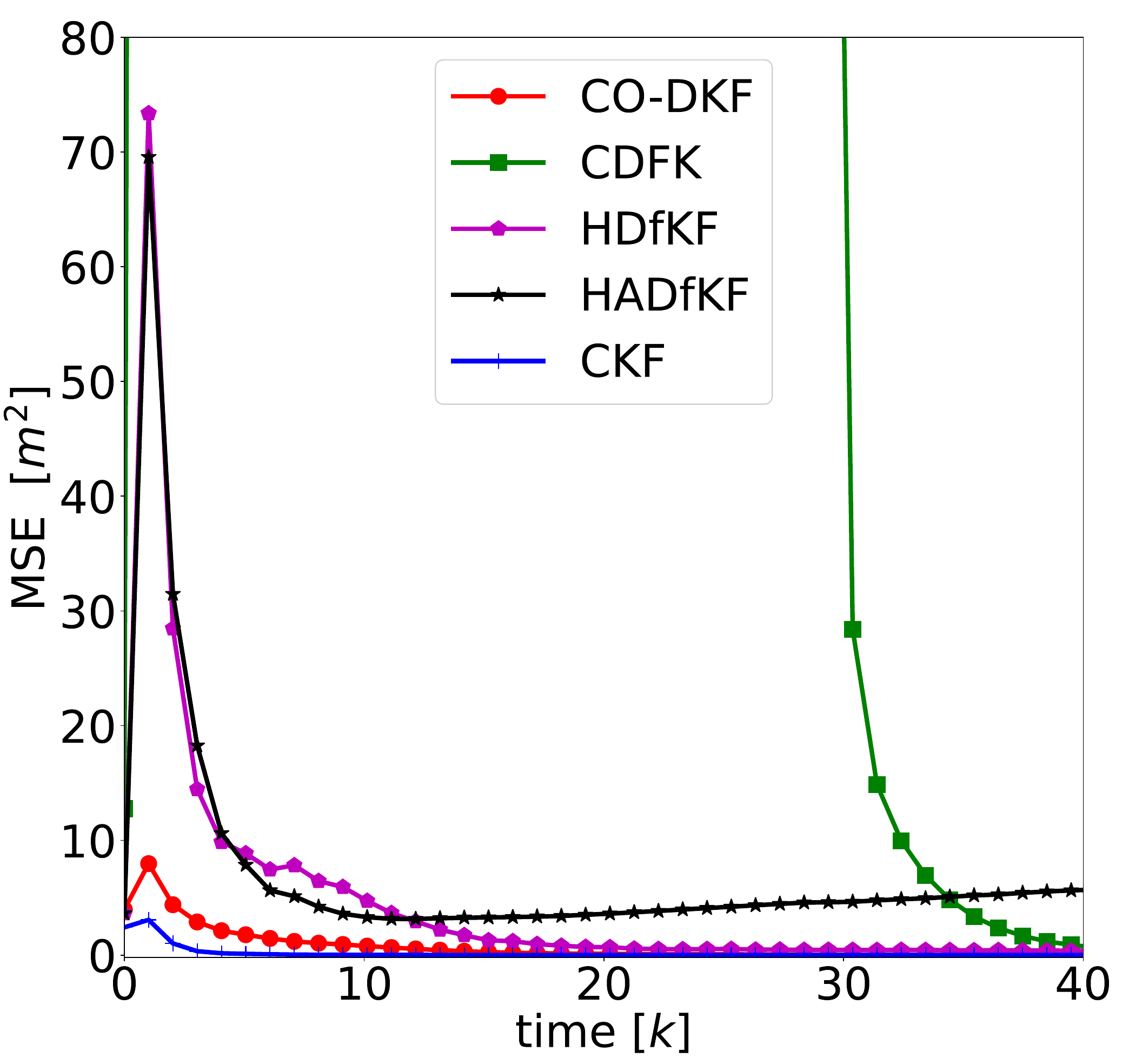}
    \\
    \includegraphics[width=0.45\columnwidth,height=0.4\columnwidth]{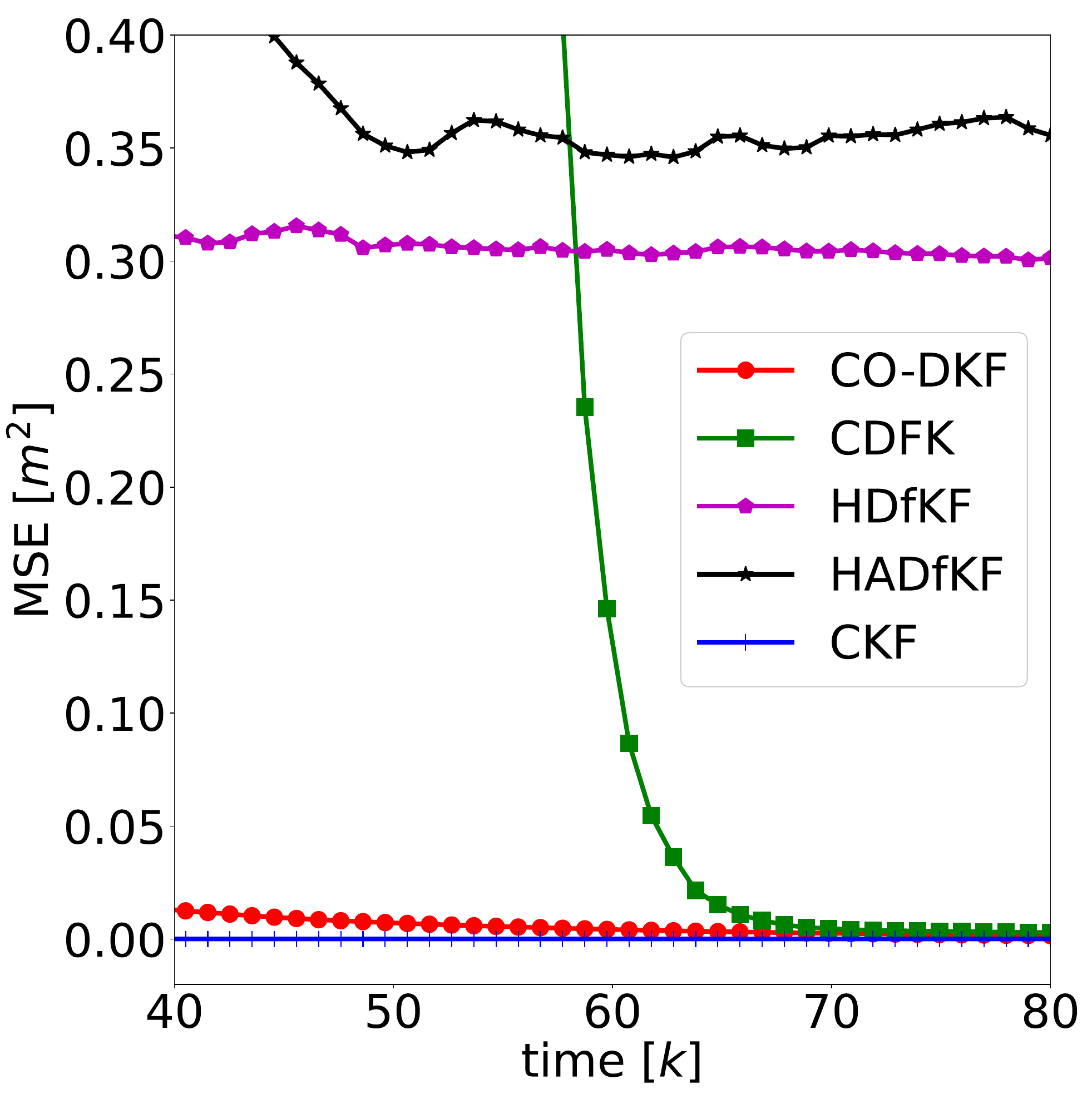}
    &
    \includegraphics[width=0.45\columnwidth,height=0.4\columnwidth]{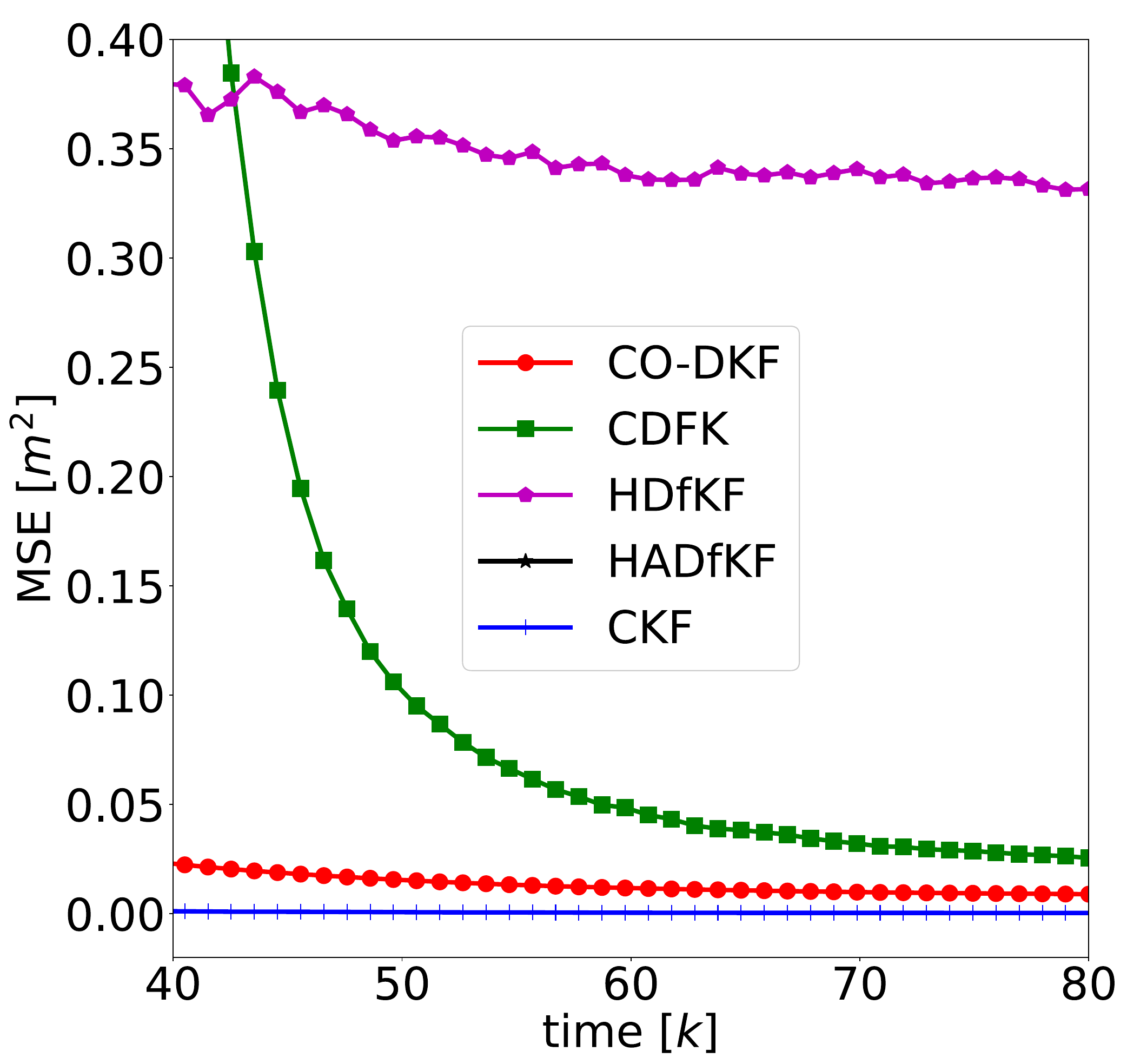}
    \\
    \includegraphics[width=0.45\columnwidth,height=0.4\columnwidth]{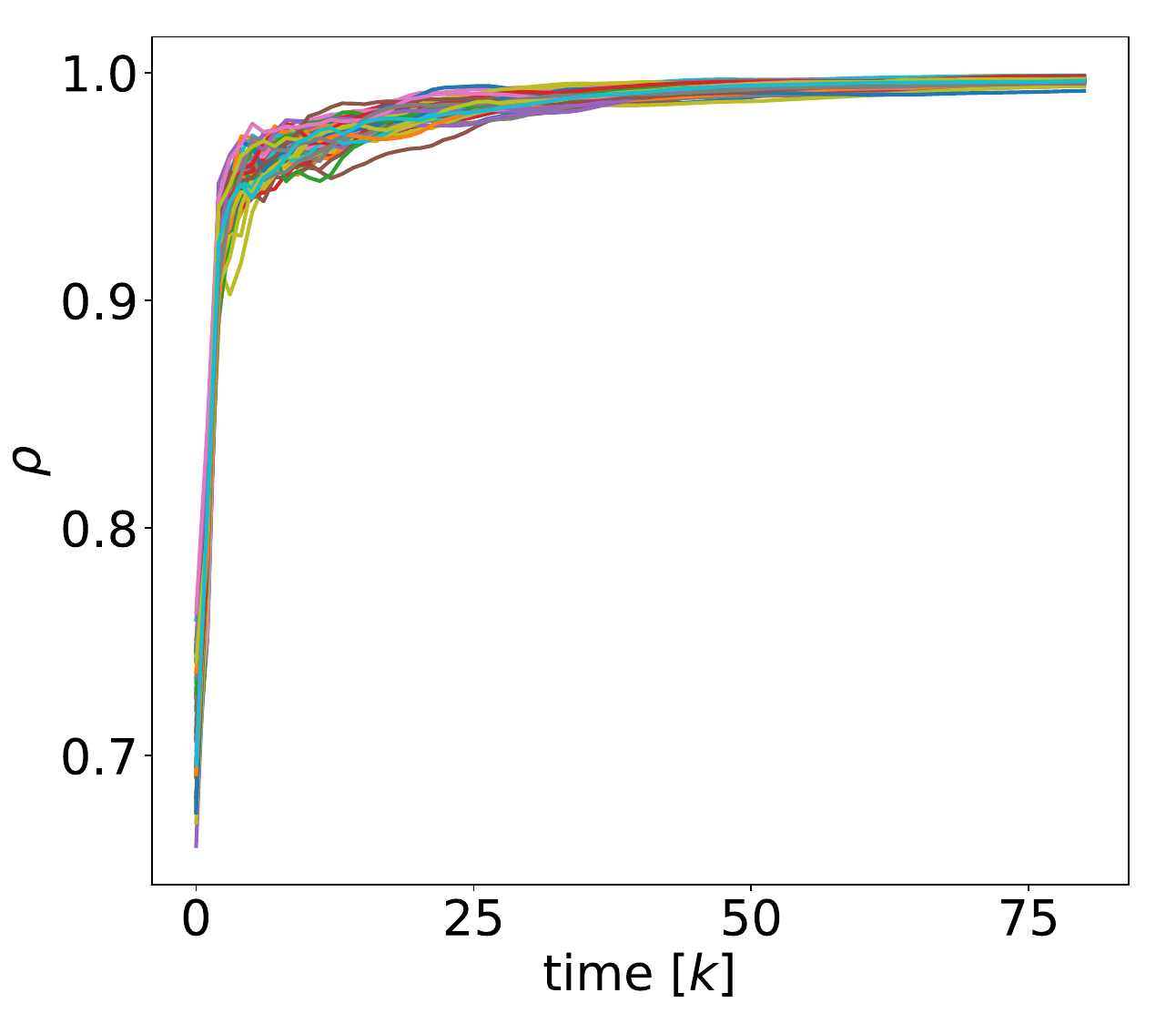}
    &
    \includegraphics[width=0.45\columnwidth,height=0.4\columnwidth]{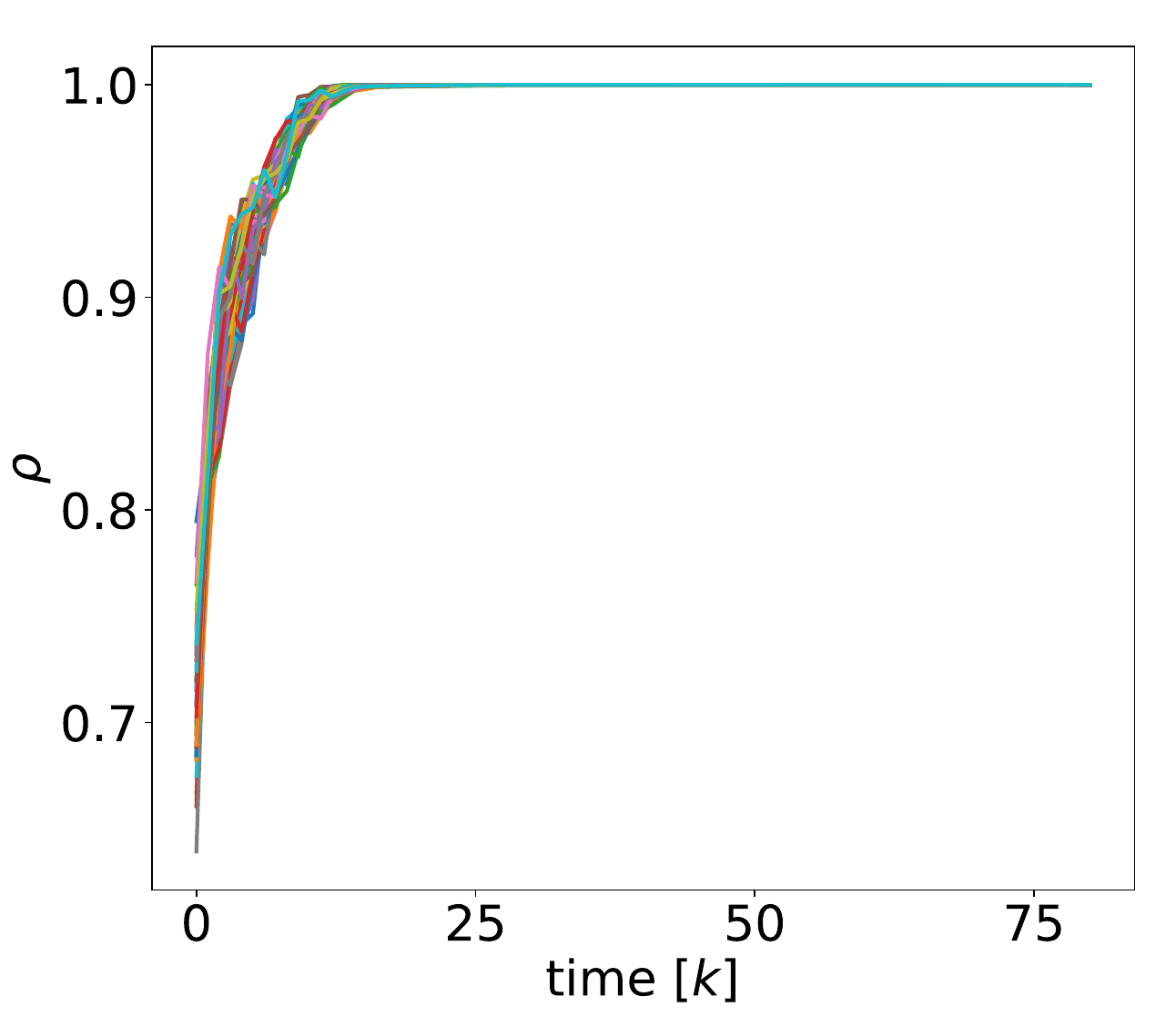}
\end{tabular}
\caption{Simulation results. First row depicts the MSE for the different estimators during the transient. Second row is a zoom of the previous in the proximity of steady-state. The third row shows the certification results, plotting the ratio achieved by the sensors over time and assigning an arbitrary colour to each sensor. CO-DKF is successful in certifying the estimation in a: (left) $96.062\%$, (right)  $93.473\%$.}
\label{fig:RMSEvsREAL_transient}
\end{figure}

The last row of Fig.~\ref{fig:RMSEvsREAL_transient} shows the value of $\rho$ in both Experiments. We see that the solution of~\eqref{eq:outerLJellipse} rapidly converges to $1$. Since in Experiment $1$ there are more high-quality sensors, the transient converges quicker than in Experiment $2$. The lower values of $\rho$ are at $0.65$, and this is only in the initialisation. We have also calculated the percentage of times that the nodes can certificate Problem~\ref{problem:intersection}, i.e.,  Constraint~\eqref{eq:constraint43} holds. This is a $96.062\%$ in Experiment $1$ and $93.473\%$ in Experiment $2$, so, along with the rapid convergence towards $\rho=1$, we can assure that our SDP relaxation is \textit{tight} and therefore CO-DKF optimality is tight.

The computation of~\eqref{eq:outerLJellipse} and~\eqref{eq:traceRelax} takes about $50$ms using CVXPY~\cite{cvxpy}, MOSEK~\cite{mosek} and an standard laptop, which can be speed up by an ad hoc implementation and consistent with $T=100$ms.


\section{Conclusions}\label{sec:conclusion}

This paper has presented CO-DKF, the first certifiable optimal distributed Kalman Filter under unknown correlations. The use of the outer Löwner-John ellipsoid method allows us to fusion the neighbouring estimates with certifiable guarantees of optimality, which comes from a tight relaxation of the original SDP problem. The optimisation is integrated in an information DKF, achieving a robust CO-DKF over heterogeneous sensor networks, with minimal communication burden and no tuning. We have also proved global asymptotic stability of the estimator and the equivalence between CO-DKF and a consensus DKF with optimal gains in the MSE sense. Two experiments have validated all the aforementioned features, returning remarkable results in sparse, highly noisy scenarios.

\balance


\bibliographystyle{IEEEtran}
\bibliography{biblio}

\end{document}